\documentclass[journal]{IEEEtran}

\usepackage{array}

\usepackage{cite}
\usepackage{amsmath,amssymb,amsfonts}
\usepackage{algorithmic}
\usepackage{graphicx}
\usepackage{textcomp}
\usepackage{xcolor}
\usepackage{amsthm}
\usepackage{multicol}

\usepackage{stfloats}
\usepackage{url}
\usepackage{verbatim}
\usepackage{algorithm}
\usepackage{enumitem}
\usepackage[caption=false,font=footnotesize]{subfig}

\newtheorem{theorem}{Theorem}

\begin{document}

\title{Selection of the Speed Command Distance for Improved Performance of a Rule-Based VSL and Lane Change Control
\thanks{This work has been supported by the METRANS Transportation Center under the following grants: Pacific Southwest Region 9 University Transportation Center (USDOT/Caltrans), the National Center for Sustainable Transportation (USDOT/Caltrans) and the South Coast AQMD.}}

\author{Tianchen Yuan, Faisal Alasiri, and Petros A. Ioannou~\IEEEmembership{Fellow,~IEEE}
\thanks{T. Yuan, F. Alasiri and P. Ioannou are with the Department of Electrical Engineering, University of Southern California, Los Angeles, CA, 90089 USA. {\tt\small E-mail: tianchey@usc.edu, alasiri@usc.edu, ioannou@usc.edu}}}

\maketitle

\begin{abstract}
Variable Speed Limit (VSL) control has been one of the most popular techniques with the potential of smoothing traffic flow, maximizing throughput at bottlenecks, and improving mobility and safety. Despite the substantial research efforts in the application of VSL control, few studies have looked into the effect of the VSL sign distance from the point of an accident or a bottleneck. In this paper, we show that this distance has a significant impact on the effectiveness and performance of VSL control. We propose a rule-based VSL strategy that matches the outflow of the upstream VSL zone with the bottleneck capacity based on a multi-section Cell Transmission Model (CTM). Then, we consider the distance of the upstream VSL zone as a control variable and perform a comprehensive analysis of its impact on the performance of the closed-loop traffic control system based on the multi-section CTM. We develop a lower bound that this distance needs to satisfy in order to guarantee homogeneous traffic density across sections and reduce bottleneck congestion. The bound is verified analytically and demonstrated using microscopic simulation of traffic on I-710 in Southern California. The simulations are used to quantify the benefits on mobility, safety and emissions obtained by selecting the upstream VSL zone distance to satisfy the analytical lower bound. The developed lower bound is a design tool which can be used to tune and improve the performance of VSL controllers. 
\end{abstract}


\section{Introduction}

\IEEEPARstart{F}{reeway} bottlenecks often caused by lane(s) drop, ramp merging or slow vehicles negatively affect traffic mobility and safety. Traffic flow control techniques such as variable speed limit (VSL) control, lane change (LC) control and ramp metering (RM), have been proposed as effective approaches to mitigate bottleneck congestion \cite{zhang2018integrated,zhang2017combined,jin2015control,li2017reinforcement,guo2020integrated,frejo2020logic}. 

The VSL controller regulates the mainstream traffic flow via speed limit commands in order to protect a freeway section from becoming congested by maximizing its throughput. The benefits of VSL control on traffic mobility have been verified by many past research efforts primarily via macroscopic simulations \cite{hegyi2008specialist,carlson2010optimal,frejo2019macroscopic}. However, inconsistent improvements have also been reported for different traffic scenarios in microscopic simulations and field tests \cite{kejun2008model,hadiuzzaman2013cell,kwon2007development}. According to the author's observations, there are four factors that potentially lead to the degradation of performance in microscopic results: the capacity drop phenomenon due to forced lane changes at the vicinity of bottlenecks, the uncertainties not captured by macroscopic models, non-optimal VSL sign locations and the shockwaves created by speed limit commands of VSL.

The capacity drop phenomenon refers to the decrease of maximum throughput at freeway bottlenecks due to forced lane change maneuvers when the traffic density is higher than the critical value \cite{chung2007relation}. It is not included in most macroscopic models but has a great influence on the traffic mobility in microscopic simulations \cite{zhang2017combined}. To address the issue, a lane change (LC) controller can be implemented at the upstream area of the bottleneck to guide the vehicles onto open lanes early, and thus, decrease the forced lane changes and significantly reduce the capacity drop phenomenon at the bottleneck\cite{zhang2017combined,guo2020integrated}. 

Many existing VSL algorithms rely on accurate measurements of traffic states to achieve the optimal performance, which can be hardly guaranteed in reality \cite{cheung2005traffic}. Model parameters such as the road capacity and backpropagation speed include uncertainties since they are typically chosen based on simulations or empirical values\cite{kontorinaki2017first}. Although robust solutions have been proposed in \cite{alasiri2020robust,yuan2021evaluation}, slightly larger size uncertainties or disturbances may still deteriorate the convergence of the closed-loop system.

The distance between VSL signs is a crucial control parameter that has been neglected by most researchers over the years \cite{martinez2020optimal}. According to \cite{carlson2010optimal}, the VSL zone distance and the discharging zone distance should be long enough for the vehicles to complete the decelerating and accelerating process respectively. In \cite{yuan2021evaluation}, the length of the upstream VSL zone is determined based on simulations. Both studies indicate that there exists a lower bound of the VSL sign distance, but no analysis has been performed on the derivation of this lower bound, which is of great significance as it can be a design parameter to improve performance.  

In a multi-section freeway model, the large reduction of VSL commands in two successive sections is another important factor that deteriorates performance as slowing down vehicles dramatically creates shockwaves due to the stop-and-go effect \cite{stern2018dissipation}. This effect is not captured in macroscopic models either. In \cite{hegyi2008specialist}, Hegyi et al. presented a VSL controller to suppress shockwaves. However, it is not designed for solving the bottleneck congestion. 

The purpose of this paper is to provide a solution to the aforementioned problems with consistent microscopic performance and real-world feasibility. The contributions of this paper in relation to our past work in \cite{yuan2021evaluation} are as follows:
\begin{enumerate}
    \item We present a much simpler and computationally more efficient rule-based VSL strategy to regulate the upstream traffic inflows to match with the bottleneck throughput and minimizes the speed variations between all the downstream sections, which is more robust against uncertainties in measurements and produces less shockwaves compared with the feedback-based VSL strategy in \cite{yuan2021evaluation}.
    \item We analytically developed a lower bound for the length of the most upstream VSL section that leads to a faster convergence to steady state densities and achieves much better benefits compared to adhoc locations used in past papers. The generated lower bound is an effective design tool in tuning and improving the performance of VSL.
\end{enumerate}

The rest of the paper is organized as follows: section \ref{section:litrev} reviews relevant literature. Section \ref{section:MSCTM} introduces the traffic model. Section \ref{section:CtrlDesign} presents a rule-based VSL strategy and analyzes the effect of the upstream VSL zone distance $L_0$. The proposed VSL controller is then combined with a LC controller to reduce the capacity drop. In section \ref{section:NumericalSimulations}, the evaluation results of multiple choices of $L_0$ under both heavy and moderate traffic conditions are presented first. Then a comparison between the proposed VSL controller and a feedback-based VSL controller is performed. Section \ref{section:Conclusion} presents the conclusion and future works.

\section{Literature Review} \label{section:litrev}
Variable speed limit (VSL) is one of the most widely studied traffic management strategies thanks to its effectiveness in regulating traffic flow and improving safety \cite{khondaker2015variable,zhang2018stability,abdel2006evaluation,yu2014optimal}. Some early VSL control studies focused on reducing the speed variations and stabilizing the traffic flow using reactive rule-based logic \cite{smulders1990control,zackor1991speed}. The improvement achieved by such rule-based VSL control approaches is often insignificant because of the limited VSL actions and the time lag between these actions. In the past two decades, the majority of the VSL control strategies were developed based on either local feedback  \cite{carlson2011local,jin2015control,zhang2018stability} or optimal control techniques \cite{hegyi2003mpc,khondaker2015variable,li2017reinforcement}. The main idea of the feedback-based VSL controller is to compute the VSL commands using the current and past traffic states, which usually requires less computation time than the optimal-control-based approach. However, the performance of the feedback-based VSL relies heavily on the accurate measurements of the traffic states, such as traffic flows and densities. Therefore, a small disturbance in measured densities, for example, may result in an unsatisfactory performance of the closed-loop system \cite{yuan2021evaluation}. The optimal-control-based VSL strategies are typically implemented within the Model Predictive Control (MPC) framework. At each time step, the VSL commands are calculated by solving an optimization problem with an objective function involving performance measures, such as total travel time (TTT), safety measurements, emission, and fuel consumption. This approach, however, does not guarantee the stability of the closed-loop system and takes substantial computational efforts when the road network is large \cite{zhang2018comparison}.

Most of the studies mentioned earlier assume a static environment with perfect measurements and models, which is not usually true in real-world scenarios. Therefore, the robustness issue of a developed VSL control against various types of uncertainties has to be examined. The existing approaches to enhance the robustness of VSL are mainly two folds: modifying classic traffic models such as the Lighthill-Whitham-Richards (LWR) model \cite{lighthill1955kinematic,richards1956shock} and the Cell Transmission Model (CTM) \cite{daganzo1994cell} to accommodate uncertainty terms; implementing a VSL controller that is less dependent on potential uncertainties. The first idea is adopted in the following studies: Liu et al. proposed a two-stage stochastic model that considers random traffic demands \cite{liu2021two}; Alasiri et al. modeled the uncertainties as an additional term in the traffic conservation law \cite{alasiri2020robust}. In accordance with the second idea, Frejo and Schutter presented a rule-based VSL controller that activates or deactivates the speed limits when the density of the corresponding bottleneck reaches a threshold that is determined offline \cite{frejo2018spert}. This approach is less dependent on the accurate measurements of traffic states. 

Another important new research direction in the application of VSL control is the effect of VSL sign locations. The VSL sign locations are often chosen empirically or based on the road configuration. This may lead to non-optimal control performance. In \cite{seraj2016optimal}, Seraj et al. suggested that the optimal VSL sign location could be determined by the space required for the traffic to reach the bottleneck capacity. In \cite{xu2019procedure}, the VSL signs were placed based on the collision risk at freeway recurrent bottlenecks. Although these studies provided some insights on the 'optimal' VSL sign location, there is no rigorous analysis that explains the impact of the VSL sign location on the performance of the VSL control. An exception is the work by Martinez and Jin \cite{martinez2020optimal}, where the authors treated the distance of the discharging zone as a control variable and optimized it based on the bottleneck capacity. The result indicates that higher speed limit commands lead to larger optimal discharging distance. To the best of our knowledge, the optimal distance of the upstream VSL zone and its impact on the closed-loop performance is still an open topic.  

The approach of integrating VSL control with other traffic flow regulation methods, such as ramp metering (RM) and lane change (LC), has been demonstrated to be more beneficial than implementing VSL solely in various scenarios \cite{carlson2010optimal,zhang2017combined,roncoli2016hierarchical,zhang2018integrated}. Many integrated approaches use optimal control or MPC to coordinate different controllers and optimize a cost function \cite{carlson2010optimal,roncoli2016hierarchical,guo2020integrated}. In \cite{zhang2018comparison}, the authors showed that a feedback-based scheme performs no worse than MPC with less computational efforts in terms of combining VSL with LC. In \cite{frejo2020logic}, Frejo and Schutter proposed a logic-based scheme that regulates the flow rates using VSL and RM controls. This is done by estimating the number of vehicles to be held or released in order to match the bottleneck flow with the capacity. The authors concluded that a well-tuned easy-to-implement integrated controller delivered similar performance, compared to an optimal controller, and thus, more suitable for real-world implementations.

\section{Multi-Section Cell Transmission Model} \label{section:MSCTM}

The Cell Transmission Model (CTM) is a discrete approximation of the Lighthill-Whitham-Richards (LWR) kinematic wave model of traffic flow \cite{lighthill1955kinematic,richards1956shock,daganzo1994cell,daganzo1995cell}. It has been widely used for traffic flow modeling and control design thanks to its simplicity and accurate ability to describe traffic dynamics at a macroscopic level \cite{kontorinaki2017first,zhang2018stability}. In the CTM framework, a highway segment is partitioned into $N$ small  homogeneous sections/cells and consecutively numbered from $1$ to $N$ in the traffic flow direction. Each section/cell is characterized by the vehicle density, inflow, and outflow, denoted by $\rho_i$, $q_i$, and $q_{i+1},$  respectively, where $i = 1, 2, \dots N.$ The density is updated using a first-order ordinary differential equation based on the traffic flow conservation, where the inflow and outflow are determined by the supply (or receiving) and demand (or sending) functions, which define a flow-density relationship known as the fundamental diagram \cite{drake1967statistical}. 

Though the original form of the CTM can reproduce traffic dynamics under both uncongested and congested conditions, it does not capture more complex traffic flow phenomena, such as the capacity drop and bounded acceleration effects due to forced lane-changing maneuvers at congested freeway bottlenecks  \cite{banks1991two,laval2006lane,zhang2017combined}. Therefore, the original form of the CTM, proposed by Daganzo in 1994, has been modified over the years in order to be consistent with the microscopic traffic flow observations \cite{zhang2017combined,jin2015control,srivastava2015modified,kontorinaki2017first,han2017resolving}. 

In this work, the most updated multi-section CTM, which takes into account the effect of both capacity drop and bounded acceleration, is considered \cite{zhang2018stability}. Without loss of generality, it is assumed that the geometry of all the sections is identical. Accordingly, the evolution of the vehicle density $\rho_i$ in each section is described by the following equations: 
\begin{equation} \label{eq:ctm_rho_i}
    \dot{\rho}_{i} = \frac{1}{L}\left(q_{i} - q_{i+1}\right), \quad \text{for} \ i = 1,...,N,
\end{equation}
where
\begin{equation}
    \label{eq:ctm_q_i}
    \begin{aligned}
        q_1 &= \min \big\{d, C, w(\rho^j - \rho_1)\big\}, \\
        q_i &= \min \big\{v_f \rho_{i-1}, \tilde{w}(\tilde{\rho}^j - \rho_{i-1}), C, 
        \\ & \qquad \quad \ \ w(\rho^j - \rho_i)\big\}, \quad \text{for} \ i = 2,..., N, \\
        q_{N+1} &= \min \big\{v_f\rho_N, \tilde{w}(\tilde{\rho}^j - \rho_N), (1- \epsilon(\rho_N))C_d\big\},
    \end{aligned}
\end{equation}
and 
\[
\epsilon(\rho_N) = \left\{ 
    \begin{aligned}
    &\epsilon_0 &&\text{if } C_d < C \text{ \& } \rho_N > \frac{C_d}{v_f}\\
    &0 && \text{otherwise}
    \end{aligned},
    \right.
\]
where the parameters in (\ref{eq:ctm_rho_i}) and (\ref{eq:ctm_q_i}) are defined in Table \ref{tab:table_parameters_def}.

\begin{table*}[h]
\centering
\caption{Definition of the Model Parameters}
\label{tab:table_parameters_def}
\begin{tabular}{lll}
\hline
\textbf{Symbol}  & \textbf{Definition} & \textbf{Unit} \\ \hline
$d$              & the upstream demand wanting to enter the considered road network & veh/h\\
$C$              & the capacity of each section/cell & veh/h\\
$C_d$            & the downstream capacity      & veh/h    \\
$v_f$            &  the free flow speed         & km/h    \\
$w$             &   the back propagation speed        & km/h       \\
$\tilde{w}$     &   the rate that the outflow $q_{i+1}$ decreases with density $\rho_i$, when $\rho_i \geq \rho_c$ \cite{srivastava2016lane}.  & km/h       \\
$\rho_c$        &     the critical density of the section/cell, at which \(v_f\rho_c = w(\rho^j - \rho_c) = \tilde{w}(\tilde{\rho}^j - \rho_c) = C.\)     & veh/km      \\
$\rho^j$        &   the jam density; the highest possible density, at which the inflow $q_i = 0.$ & veh/km\\
$\tilde{\rho}^j$  &   the jam density associated with outflow $q_{i+1}.$        & veh/km         \\
$L$      &    the length of each section/cell.       & km   \\
$\epsilon_0$     &    the capacity drop factor, where $\epsilon_0 \in (0,1).$    & unitless     \\ 
\hline
\end{tabular}
\end{table*}

Based on the presented macroscopic multi-section CTM, the goal is to design a traffic flow controller so that the traffic conditions of all the road sections operate within the free-flow region of the fundamental diagram, despite the activation of the downstream bottleneck. Since the mainstream traffic flow is to be regulated, variable speed limit (VSL) control is a reasonable traffic flow control strategy. The underlying idea is to regulate the inflow, $q_1,$ to a level that is within the capacity constraints of the downstream section at the bottleneck. Furthermore, minimize the speed variations between all the CTM sections to diminish the stop-and-go traffic behavior and, thus, achieve smooth traffic flow conditions. Driven by this idea, a rule-based VSL control is proposed in the following section, considering the distance of the most upstream VSL zone, denoted by $(L_0)$, as a control variable. In addition, a lane change (LC) controller is combined with the VSL to manage the lane-changing maneuvers in the vicinity of the bottleneck in order to prevent the VSL performance from being deteriorated \cite{zhang2017combined}.  


\section{Control Design} \label{section:CtrlDesign}
This section aims to develop a combined VSL and LC control design and analyze the stability properties of the closed-loop system. The section length covered by the most upstream VSL sign is treated as a variable in the design and its impact on the performance of the closed-loop system is investigated.

\subsection{Rule-Based Variable Speed Limit Control} \label{subsec:VSL}
We propose a rule-based VSL controller to alleviate freeway bottleneck congestion, based on the multi-section CTM presented in section \ref{section:MSCTM}. The VSL control signs are implemented in the upstream of the first section as well as all CTM sections as shown in Fig. \ref{fig:L0Effect}. Each VSL command takes effect at the beginning of the section. The downstream bottleneck is created by a lane closure due to an incident. The control objective is to match the inflow of the first section, $q_1$, with the bottleneck capacity during the incident. 
\begin{figure}[h]
\centering
\includegraphics[width = 0.48\textwidth]{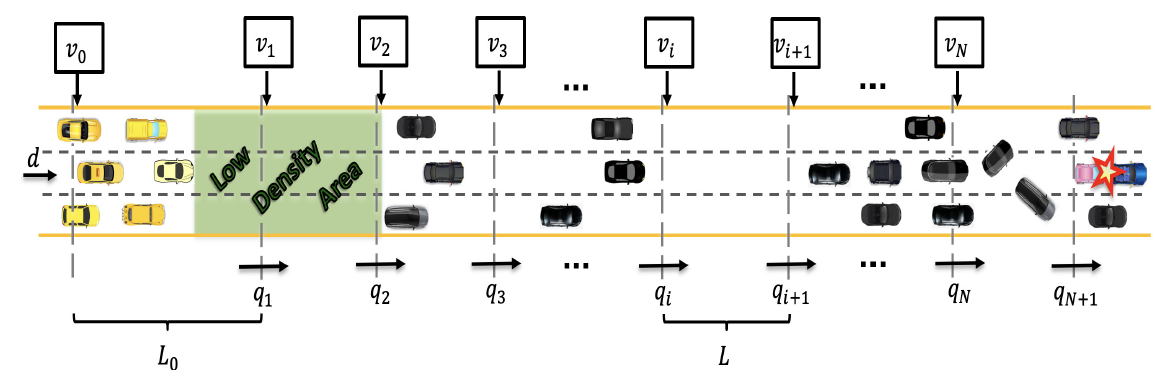}
\caption{Representation of a motorway stretch within the multi-section CTM framework with the VSL control.}
\label{fig:L0Effect}
\end{figure}

The maximum possible flow produced by a speed limit $v_i$ ($v_i \leq v_f$) can be derived from the geometry of the fundamental diagram shown in Fig. \ref{fig:FDwithVSL}, that is, $\frac{v_i w \rho^j}{v_i + w}$. Therefore, when the VSL commands take effect, the traffic flow dynamics are formulated as follows: 
\begin{equation} 
    \label{eq:qi_vsl}
    \begin{aligned}
        q_1 &= \min \big\{d,\frac{v_0 w\rho^j}{v_0 + w},\frac{v_1 w\rho^j}{v_1 + w},w(\rho^j - \rho_1)\big\}, \\
        q_i &= \min \big\{v_{i-1}\rho_{i-1},\frac{v_{i-1} w\rho^j}{v_{i-1} + w},\frac{v_i w\rho^j}{v_i + w},
        \\ & \qquad \quad \ \ w(\rho^j - \rho_i)\big\}, \quad \text{for} \ i = 2,..., N, \\
        q_{N+1} &= \min \big\{v_N\rho_N,(1-\epsilon(\rho_N))C_d,\tilde{w}(\tilde{\rho}^j-\rho_N)\big\}  
    \end{aligned}
\end{equation}

\begin{figure}[h]
\centering
\includegraphics[width = 0.45\textwidth]{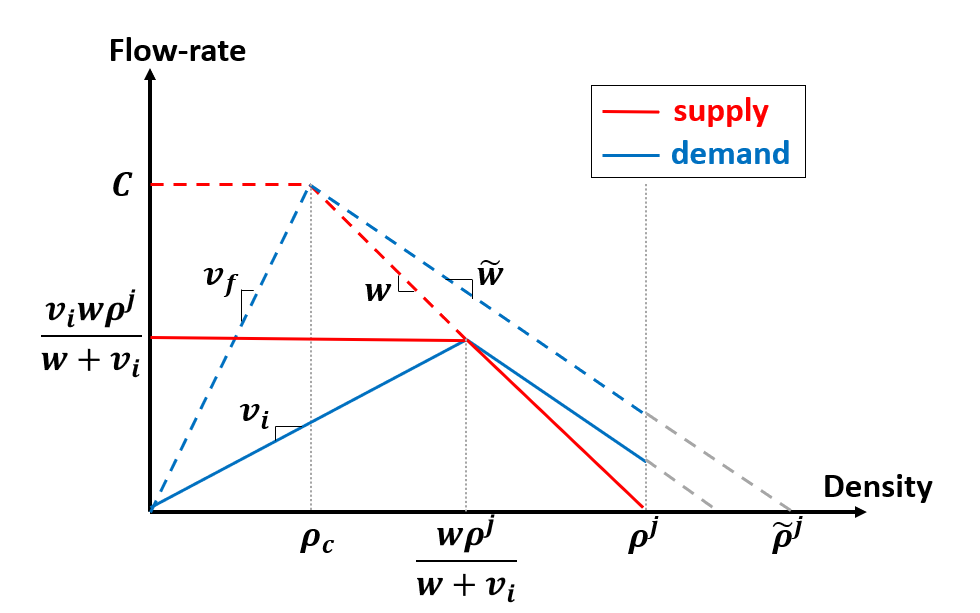}
\caption{Triangular fundamental diagram under VSL Control. The red curve represents the supply function. The blue curve represents the demand function.}
\label{fig:FDwithVSL}
\end{figure}

As discussed in section \ref{section:MSCTM}, when computing the bottleneck capacity $C_b$, we take the capacity drop phenomenon into consideration:
\begin{equation}
    \begin{aligned}
        C_b &= (1-\epsilon(\rho_N))C_d, \\
        \epsilon(\rho_N) &= \left\{ 
        \begin{aligned}
        &\epsilon_0 &&\text{if } C_d < C \text{ \& } \rho_N > \frac{C_d}{v_f}, \\
        &0 && \text{otherwise}
        \end{aligned}
        \right.
    \end{aligned}
\end{equation}

The key idea is to force the inflow $q_1$ to be less than or equal to the bottleneck capacity $C_b$ under various traffic conditions by adjusting the most upstream speed limit $v_0$. To be more specific, we consider the following three scenarios based on different levels of traffic demands $d$:
\begin{itemize}
    \item scenario 1: $d < (1-\epsilon_0)C_d$
    \item scenario 2: $(1-\epsilon_0)C_d \leq d \leq C_d$
    \item scenario 3: $d > C_d$
\end{itemize}

In scenario 1, the demand $d$ is less than the bottleneck capacity and no control effort is needed. Therefore, we simply let $v_0 = v_f$. In scenario 2, $d$ exceeds the dropped bottleneck capacity $(1-\epsilon_0)C_d$ when there exists congestion near the bottleneck. However, $d$ does not exceed the recovered bottleneck capacity $C_d$ when the congestion is cleared. The VSL control needs to be activated in order to regulate the inflow as follows:
\begin{equation}
    q_1 = 
    \left\{
    \begin{aligned}
        &d &&\text{if } 0 \leq \rho_N \leq \frac{C_d}{v_f}, \\
        &\frac{v_0 w \rho^j}{v_0 + w} = (1-\epsilon_0)C_d &&\text{otherwise} 
    \end{aligned}
    \right.
\end{equation}
In scenario 3, since $d > C_d$, the VSL control needs to be activated all the time to regulate the inflow $q_1$ so that:
\begin{equation}
    q_1 = \frac{v_0 w \rho^j}{v_0 + w} = 
    \left\{
    \begin{aligned}
        &C_d &&\text{if } 0 \leq \rho_N \leq \frac{C_d}{v_f}, \\
        &(1-\epsilon_0)C_d &&\text{otherwise} 
    \end{aligned}
    \right.
\end{equation}
Combining all three scenarios together, the VSL control command $v_0$ can be computed as:
\begin{equation} \label{eq:v0_VSLcommand}
    v_0 = 
    \left\{
    \begin{aligned}
        &\frac{w C_d}{w\rho^j-C_d} &&\text{if } d > C_d \  \& \ 0 \leq \rho_N \leq \frac{C_d}{v_f}, \\
        &\frac{w (1-\epsilon_0)C_d}{w\rho^j-(1-\epsilon_0)C_d} &&\text{if } d \geq (1-\epsilon_0)C_d \  \&\  \rho_N > \frac{C_d}{v_f}, \\
        &v_f &&\text{otherwise}
    \end{aligned}
    \right.
\end{equation}
Note that the condition $\rho_N > \frac{C_d}{v_f}$ represents the existence of bottleneck congestion, and its counterpart $0 \leq \rho_N \leq \frac{C_d}{v_f}$ means that the congestion has been removed. The bottleneck capacity recovers from $(1-\epsilon_0)C_d$ to $C_d$ during the transition from the former condition to the latter. In both scenarios 2 and 3, we switch $v_0$ from a lower value to a higher value once the transition is completed as indicated by (\ref{eq:v0_VSLcommand}) in order to maximize the bottleneck throughput.  

By matching $q_1$ upstream with $C_b$, the remaining downstream sections are set to maintain a steady traffic flow, achieved by simply letting
\begin{equation} \label{eq:vi_VSLcommand}
    v_i=v_f \quad \text{for } i=1,...,N
\end{equation}
Our previous research indicates that variations in the given speed limit commands, especially from high to low speeds, create backpropagations and, thus, deteriorate traffic mobility \cite{yuan2021evaluation}. This should be avoided as much as possible. In section \ref{section:NumericalSimulations}, we will show that during high density traffic conditions, the proposed VSL controller performs better than a classical feedback-based VSL controller, where the speed limit commands vary based on measured densities and flows.

\subsection{Analysis of VSL Zone Distance} \label{subsec:AnalysisL0}
In this subsection, we consider the length of the upstream VSL zone $L_0$ where the speed limit given by (\ref{eq:v0_VSLcommand}) will be applied as a control variable and analyze the impact of $L_0$ on the performance of the closed-loop system. According to the previous discussion, the VSL controller is activated immediately once the incident takes place, which creates a speed difference between the upstream VSL zone ($v_0$) and all the downstream sections ($v_1,...,v_N$) when $d \geq (1-\epsilon_0)C_d$. To demonstrate this, we use the scenario illustrated in Fig. \ref{fig:L0Effect}. The black vehicles entered the road network before the activation of the VSL. They travel without any restrictions from the speed limit signs and slow down as they approach the congested area in front of the incident location. The yellow vehicles entered the road network after the activation of the VSL controller. They travel at a reduced speed $v_0$ within the upstream VSL zone. As a result, a low-density area is created between the two groups of vehicles, allowing the shockwaves from the bottleneck to be absorbed. The existence and propagation of the low-density area can be observed from the flow curves shown in Fig. \ref{fig:flow_lowdensity}. Note that the vehicle input flow never drops down to 0, and the inflow drops in all downstream sections because of the VSL control.

\begin{figure}[h]
    \centering
    \includegraphics[width = 0.45\textwidth]{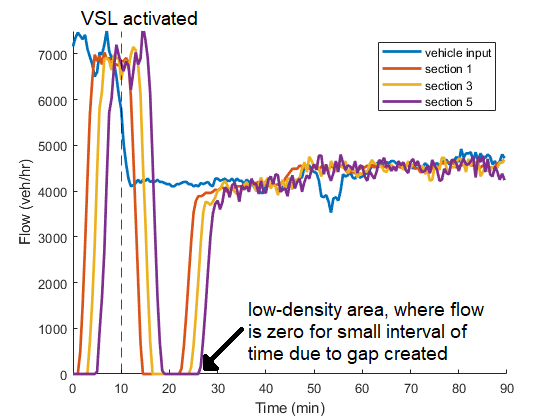}
    \caption{The behavior of flow curves with time when the VSL is activated.}
    \label{fig:flow_lowdensity}
\end{figure}

The low-density area needs to be long enough so that the congestion associated with the black vehicles is cleared before the yellow vehicles catch up. In other words, the first yellow vehicle should never reach the last black vehicle within the road network as it may create further shockwaves. This can be formulated as a chasing problem in which the time it takes to clear congestion at the bottleneck, denoted as $T_b$, is strictly less than the time spent for the first yellow vehicle to reach the bottleneck, denoted as $T_y$, i.e.,  $T_b < T_y$. 

\begin{theorem} \label{thm:bound_L0}
    Consider the freeway bottleneck control problem with a constant demand $d \geq (1-\epsilon_0)C_d$ and VSL commands given by (\ref{eq:v0_VSLcommand}) and (\ref{eq:vi_VSLcommand}). The propagation of the traffic congestion at the bottleneck can be completely absorbed by the low-density area created by the VSL control if the upstream VSL zone distance $L_0$ satisfies
    \begin{equation} \label{eq:bound_L0}
        L_0 > \frac{(v_f \sum_{i=1}^{N} \rho_i(t_0) - (1-\epsilon_0) C_d N)v_0 L}{((1-\epsilon_0)C_d - v_0\rho_0(t_0))v_f}
    \end{equation}
    where $\rho_i$ is the measured density of section $i$, for $i=0,1,2,...,N$. $C_d$ is the downstream capacity; $v_f$ is the free flow speed; $v_0$ is the upstream VSL command; $\epsilon_0$ is the capacity drop factor; $N$ is the number of downstream sections; $L$ is the length of each downstream section and $t_0 \geq 0$ is the time the incident takes place.
\end{theorem}
\begin{proof}
    Let's start by estimating the number of vehicles $N_b$ that already entered the network (colored in black in Fig. \ref{fig:L0Effect}) using the measured densities at the time $t_0$:
    \begin{equation} \label{eq:est_Nb}
        N_b = L_0\rho_0(t_0) + L\sum_{i=1}^{N} \rho_i(t_0)
    \end{equation}
    Since the traffic flow at the bottleneck is $q_{N+1}$, we have
    \begin{equation} \label{eq:est_Tb_1}
        \int_{t_0}^{t_0 + T_b} q_{N+1}(\tau)d\tau = N_b
    \end{equation}
    When congestion is active near the bottleneck, $q_{N+1}$ is equal to $(1-\epsilon_0)C_d$ due to the capacity drop phenomenon. Thus,
    \begin{equation} \label{eq:est_Tb_2}
        \int_{t_0}^{t_0 + T_b} q_{N+1}(\tau)d\tau = (1-\epsilon_0)C_d T_b
    \end{equation}
    From (\ref{eq:est_Nb}),(\ref{eq:est_Tb_1}) and (\ref{eq:est_Tb_2}), we obtain the time to clear congestion at the bottleneck as
    \begin{equation} \label{eq:Tb}
        T_b = \frac{L_0\rho_0(t_0) + L\sum_{i=1}^{N} \rho_i(t_0)}{(1-\epsilon_0)C_d}
    \end{equation}
    On the other hand, the yellow vehicles are able to follow the speed limit commands when traveling through the network. The time the first yellow vehicle reaches the bottleneck denoted by $T_y$ can be computed as
    \begin{equation} \label{eq:Ty}
        T_y = \frac{L_0}{v_0}+\frac{NL}{v_f}
    \end{equation}
    Therefore, $T_b<T_y$ yields
    \begin{equation} \label{eq:ineq_L0v0}
        \frac{L_0\rho_0(t_0) + L\sum_{i=1}^{N} \rho_i(t_0)}{(1-\epsilon_0)C_d} < \frac{L_0}{v_0}+\frac{NL}{v_f}
    \end{equation}
    which is equivalent to (\ref{eq:bound_L0}) giving the following condition:
    \begin{equation} \label{eq:v0_loosebound}
        v_0 < \frac{(1-\epsilon_0)C_d}{\rho_0(t_0)}
    \end{equation}
    Note that (\ref{eq:v0_loosebound}) is automatically satisfied as we compute $v_0$ with (\ref{eq:v0_VSLcommand}). Therefore,  (\ref{eq:bound_L0}) guarantees that $T_b<T_y$, which prevents shockwaves and improves the performance of the closed-loop system.
\end{proof}

According to (\ref{eq:bound_L0}), the lower bound of $L_0$ is positively correlated with the speed limit of the upstream VSL zone, $v_0$, and the initial densities of all sections, $\rho_i(t_0)$, for $i=0,...,N$. Although there is no theoretical upper bound on $L_0$, overextending it leads to undesirable travel time and underutilization of the road. We will demonstrate the impact of $L_0$ and verify the effectiveness of the computed lower bound under different traffic scenarios via microscopic simulations in section \ref{section:NumericalSimulations}.

\subsection{Lane Change Control} \label{subsec:LCC}
According to (\ref{eq:Tb}), the time spent for clearing the bottleneck, i.e., $T_b$, can be reduced by increasing the bottleneck throughput. Therefore, we implement the Lane Change (LC) control in the discharging section in order to reduce the capacity drop, increase the bottleneck throughput, and accelerate the convergence. 

The mechanism of LC control involves two ingredients. The first one is to give appropriate lane-changing recommendations to vehicles moving in the closed lane before approaching the bottleneck. The second ingredient is determining at what distance from the bottleneck, referred to as $d_{LC}$, these recommendations are provided. $d_{LC}$ needs to be long enough so that the vehicles can complete the lane change maneuvers safely, but an overextended $d_{LC}$ may lead to the underutilization of the road. In \cite{zhang2017combined}, an empirical formula is proposed to determine the value of $d_{LC}$ as follows:
\begin{equation}
    d_{LC} = \xi \cdot n
\end{equation}
where $n$ is the number of lanes closed at the bottleneck, $\xi$ is a design parameter that depends on the traffic demand. $\xi$ could be found by simulations. For the specific highway segment in our case, $\xi$ takes the value of 800 m and 700 m for high and moderate traffic demands respectively \cite{zhang2017combined}.

\subsection{Stability Analysis}
In this section, we perform a rigorous stability analysis of the closed-loop system with the proposed integrated VSL and LC controller. There are three control variables: the upstream VSL command $v_0$ given by (\ref{eq:v0_VSLcommand}), the upstream VSL zone distance $L_0$ whose lower bound is given by (\ref{eq:bound_L0}), and the lane change distance $d_{LC}$ suggested by \cite{zhang2017combined}. Accordingly, the closed-loop system (\ref{eq:ctm_rho_i})-(\ref{eq:qi_vsl}) with (\ref{eq:v0_VSLcommand}) can be expressed as follows:
\begin{equation} \label{eq:CLsys_rho_i}
    \dot{\rho}_{i} = \frac{1}{L} \left(q_{i} - q_{i+1} \right), \quad \text{for} \  i = 1,...,N,
\end{equation}
where
\begin{equation} 
    \label{eq:CLsys_qi}
    \begin{aligned}
        q_1 &= \min \big\{d,\frac{v_0 w\rho^j}{v_0 + w},C,w(\rho^j - \rho_1)\big\},\\
        q_i &= \min \big\{v_f\rho_{i-1},C,w(\rho^j - \rho_i)\big\}, \quad \text{for } i = 2,..., N, \\
        q_{N+1} &= \min \big\{v_f\rho_N,(1-\epsilon(\rho_N))C_d,\tilde{w}(\tilde{\rho}^j-\rho_N)\big\}
    \end{aligned}
\end{equation}

Based on the analysis in section \ref{subsec:AnalysisL0}, vehicles that enter the network before the activation of VSL control (colored in black in Fig. \ref{fig:L0Effect}) travel downstream freely and create congestion at the bottleneck. Although the traveling speed is difficult to compute when the traffic is congested, the bottleneck throughput is stabilized around $(1-\epsilon_0)C_d$, meaning that the congested density is at steady state for the given demand. According to theorem \ref{thm:bound_L0}, the congestion can be removed without affecting the new group of vehicles (held by $v_0$, colored in yellow in Fig. \ref{fig:L0Effect}) by choosing $L_0$ appropriately. Once these (yellow) vehicles exit the upstream VSL zone, they travel through all downstream sections in free-flow speed. Note that the speed is maintained even when they approach the bottleneck because we have matched the inflow $q_1$ with the bottleneck capacity and applied LC recommendations to avoid sudden lane changes. As a result, system (\ref{eq:CLsys_rho_i})-(\ref{eq:CLsys_qi}) can be further simplified when the congestion is removed ($t > t_0 + T_b$) as follows:

\begin{equation}
    \label{eq:CLsys_simplified_rho_i}
    \begin{aligned}
        \dot{\rho}_{1} &= \frac{1}{L}(\min \{d,\frac{v_0 w\rho^j}{v_0 + w}\} - v_f\rho_{1}), \\
        \dot{\rho}_{i} &= \frac{v_f}{L}(\rho_{i-1} - \rho_{i}), \quad \text{for} \  i = 2,...,N \ 
    \end{aligned}
\end{equation}

\begin{theorem} \label{theorem:stability_CLsys}
    Consider the closed-loop system (\ref{eq:CLsys_rho_i})-(\ref{eq:CLsys_qi}) with the traffic demand $d \geq (1-\epsilon_0)C_d$. The proposed VSL controller given by (\ref{eq:v0_VSLcommand})-(\ref{eq:bound_L0}) guarantees that the density converges exponentially fast to the equilibrium point $\rho_i^e = \frac{\min\{d,C_d\}}{v_f}$ ($i=1,...,N$) for $t > t_0 + T_b$ which corresponds to the maximum possible throughput, where $t_0$ is the time the controller is activated and $T_b$ is given by (\ref{eq:Tb}).
\end{theorem}
\begin{proof}
Let us define $\rho=[\rho_1,\rho_2,...,\rho_N]^\text{T}$, and then compute the equilibrium point of system (\ref{eq:CLsys_simplified_rho_i}) by setting $\dot{\rho}=0$, which yields
\begin{equation}
    \label{eq:CLsys_equilibrium_1}
        \rho_1^e = \rho_2^e = ... = \rho_N^e = \frac{\min \{d,C_d\}}{v_f}, \quad t > t_0 + T_b
\end{equation}

Note that $v_0$ has already been switched to the higher value after the congestion is cleared according to (\ref{eq:v0_VSLcommand}), which corresponds to the maximum bottleneck throughput $\min \{d,C_d\}$ without the capacity drop. To prove the exponential convergence, we rewrite system (\ref{eq:CLsys_simplified_rho_i}) in matrix form:
\begin{equation}
    \label{eq:CLsys_simplified_rho_i_matrix}
    \dot{\rho} = \frac{1}{L}(A_1\rho + b_1), \quad t > t_0 + T_b \\
\end{equation}
where
\[
  A_1 =
  \begin{bmatrix}
    -v_f & 0 & \dots & 0 & 0 \\
    v_f & -v_f & \dots & 0 & 0 \\
    \vdots & \ddots & \ddots & \vdots & \vdots \\
    0 & \dots & v_f & -v_f & 0 \\
    0 & \dots & 0 & v_f & -v_f
  \end{bmatrix} \ 
  b_1 = 
  \begin{bmatrix}
    \min \{d,C_d\} \\
    0 \\
    \vdots \\
    0 \\
    0
  \end{bmatrix}
\]
$A_1$ is a lower triangular matrix with all diagonal entries being negative real numbers. Therefore, system (\ref{eq:CLsys_simplified_rho_i_matrix}) is exponentially stable when $t > t_0 + T_b$.
\end{proof}


\section{Numerical Simulations} \label{section:NumericalSimulations}

\subsection{Simulation Network and Parameter Selection}
The commercial microscopic simulator PTV VISSIM 10 is used to evaluate the performance of the proposed controller. The road network in Fig. \ref{fig:I710simnet} is a simplified version of a 14.4-km (9 mi) segment of I-710 freeway from I-105 to the Long Beach Port in California, United States. There are 6 CTM sections with a unified length of 1.6 km (1 mi). The length of the upstream VSL zone is determined according to (\ref{eq:bound_L0}) and verified by the simulation results. The network has 3 lanes, and no on-ramps or off-ramps are considered. Each simulation run lasts for 90 min. At 10 min, the middle lane is closed due to an incident that takes place at the downstream exit of the road network, which creates a bottleneck when the demand $d$ is larger than the downstream capacity $C_d$. The incident is cleared eventually at 80 min. 
\begin{figure}[h]
\centering
\includegraphics[width = 0.48\textwidth]{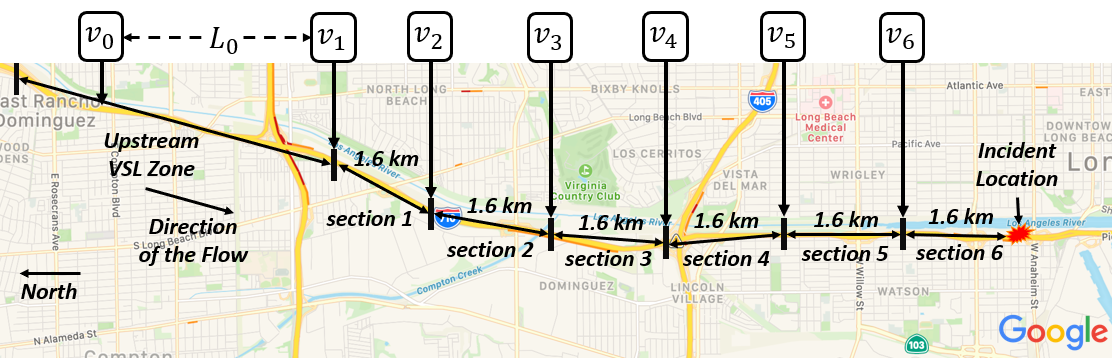}
\caption{The I-710 simulation network.}
\label{fig:I710simnet}
\end{figure}

We first run multiple simulations with the demand $d$ gradually increasing in the open-loop system (without any control). Based on the collected measurements of the flow and density of the last section, we draw the fundamental diagram of the simulation network as shown in Fig. \ref{fig:FDsimnet}. The following parameters can be determined according to the fundamental diagram: the road capacity $C=7200$ veh/h, the bottleneck capacity $C_d=4800$ veh/h, the bottleneck capacity with the capacity drop $(1-\epsilon_0)C_d=4300$ veh/h. Thus, the level of the capacity drop can be estimated as $\epsilon_0=0.1$. The free-flow speed $v_f$ is set to be 100 km/h. The backpropagation speeds are selected with the empirical values proposed in \cite{alasiri2020robust}: $w=30$ km/h and $\tilde{w}=15$ km/h. Using the geometry in Fig. \ref{fig:FDwithVSL}, we have $\rho^j = C/v_f+C/w = 312$ veh/km and $\tilde{\rho}^j = C/v_f+C/\tilde{w} = 552$ veh/km.
\begin{figure}[h]
\centering
\includegraphics[width = 0.45\textwidth]{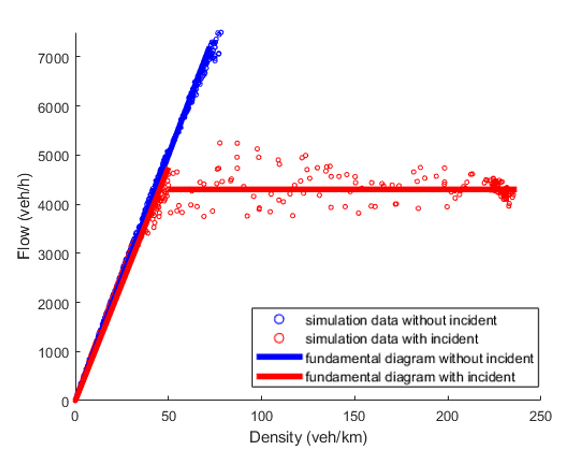}
\caption{Fundamental diagram of the simulation network with (red) and without (blue) incident.}
\label{fig:FDsimnet}
\end{figure}

\subsection{Performance Measurements}
The performance measurements and criteria that the authors used in \cite{zhang2017combined} are adopted here to evaluate the effectiveness of the proposed controller. The following are brief definitions of the performance measurements, mainly for the sake of completeness:
\begin{itemize} 
    \item Average Travel Time (ATT): the average time spent for each vehicle to travel through the whole network.
    \begin{equation}
        ATT = \frac{1}{N_v}\sum_{i=1}^{N_v} (t_{i,out}-t_{i,in})
    \end{equation}
    where $N_v$ is the number of vehicles passing through the network, $t_{i,in}$ and $t_{i,out}$ is the time vehicle $i$ enters and exits the network respectively. 
    \item Average number of stops: the average number of stops performed by each vehicle when traveling in the network.
    \begin{equation}
        \bar s = \frac{1}{N_v}\sum_{i=1}^{N_v} s_i
    \end{equation}
    where $s_i$ is the number of stops performed by vehicle $i$.
    \item Average emission rates of CO$_2$: the calculation of emission rates is based on the MOVES model provided by the Environment Protection Agency (EPA) \cite{epa2010motor}. 
    \begin{equation}
        \bar R = \sum_{i=1}^{N_v} E_i / \sum_{i=1}^{N_v} d_i
    \end{equation}
    where $E_i$ is the emission produced by vehicle $i$ and $d_i$ is the travelled distance of vehicle $i$
    \item The relative root mean square error (RRMSE): we compare the average density measurements of each downstream sections with the desired equilibrium and compute the RRMSE to indicate whether the closed-loop system is stabilized around the desired equilibrium. 
    \begin{equation}
        e_\rho = \frac{1}{\rho^*}\sqrt{\frac{1}{t_e-t_s}\int_{t_s}^{t_e}(\bar\rho(\tau)-\rho^*)^2}
    \end{equation}
    where $\rho^*=\frac{\min\{d,C_d\}}{v_f}$ is the desired equilibrium, $t_e$ is the time when the incident ends, $t_s$ is the time we switch $v_0$ to a higher value, $\bar\rho$ is the average density measurement. In our simulations, we assume that the section densities reach a steady state if $e_\rho \leq 25\%$. 
\end{itemize}

\subsection{Microscopic Simulations}
In this section, we present the simulation results under both high and moderate traffic demands where $d=7000$ and $5500$ veh/h respectively. For high demand, we evaluate 12 scenarios with the upstream VSL zone distance $L_0=[0,0.8,1.2,1.4,1.6,1.8,2.0,2.2,2.4,3.2,4.0,4.8]$ km. For moderate demand, we evaluate 9 scenarios with $L_0=[0,0.4,0.6,0.8,1.0,1.2,1.6,3.2,4.8]$ km. We take the average results of 10 independent Monte Carlo simulations for each scenario to reduce the randomness and increase reliability.

Since both high and moderate demands exceed the bottleneck capacity $C_d=4800$ veh/h, the traffic conditions before the bottleneck congestion being cleared are $d>C_d$ and $\rho_N>C_d/v_f$, which falls into the second case of (\ref{eq:v0_VSLcommand}). Plugging the model parameters, we have $v_0=25.7$ km/h. After the bottleneck congestion is removed, the traffic conditions become $d>C_d$ and $0 \leq \rho_N \leq C_d/v_f$, which falls into the first case of (\ref{eq:v0_VSLcommand}). In this case, $v_0=31.6$ km/h. Theoretically, these two values of $v_0$ change $q_1$ to match the dropped and recovered bottleneck capacity respectively. However, these values may be too aggressive in practice and should be considered as the upper bounds of $v_0$ due to the uncertainties in model parameters such as $w$ and the randomness in microscopic simulations. Therefore, we select $v_0$ to be 20 and 25 km/h before and after the removal of the bottleneck congestion. To determine the proper time of switching $v_0$ from 20 to 25 km/h, we use the time needed to clear the bottleneck congestion given (\ref{eq:Tb}) calculated to be 14 min for high demand and 11 min for moderate demand. Taking the potential uncertainties into account, the actual time required for dissipating the congestion may be longer. To enhance the robustness and ensure that the switching happens after the removal of the congestion, we set the switching time to be 20 min after the occurrence of the incident for both demand levels, i.e. $t_s=30$ min. In summary, the control command of $v_0$ for both high and moderate demands is given as 
\begin{equation} \label{eq:v0_simu}
    v_0 = 
    \left\{
    \begin{aligned}
        &20 \text{ km/h} &&\text{if } 10 \leq t < 30, \\
        &25 \text{ km/h} &&\text{if } 30 \leq t < 80, \\
        &100 \text{ km/h} &&\text{otherwise}
    \end{aligned}
    \right.
\end{equation}

Plugging the lower value of $v_0$ (20 km/h) and $d$ into (\ref{eq:bound_L0}), we obtain the lower bounds of $L_0$ as 1.8 km for high demand and 0.7 km for moderate demand. In our simulation study we vary $L_0$ from values below and above the lower bound and examine the impact on the performance of the VSL controller. 

Fig. \ref{fig:erho_L0}-\ref{fig:ATT_L0} show the evaluation results of all the above-mentioned scenarios under both traffic demands in terms of the RRMSE in densities, the average number of stops, the average emission rates of CO$_2$ and the average travel time (ATT) respectively. Each fitting curve is generated by the smoothing spline fitting algorithm in MATLAB. 

\begin{figure}[h]
\centering
\includegraphics[width = 0.45\textwidth]{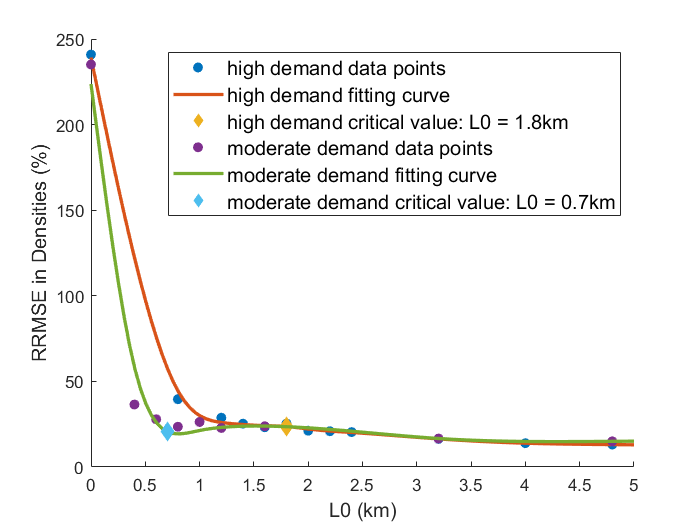}
\caption{RRMSE in densities ($e_\rho$) vs. $L_0$.}
\label{fig:erho_L0}
\end{figure}

\begin{figure}[h]
\centering
\includegraphics[width = 0.45\textwidth]{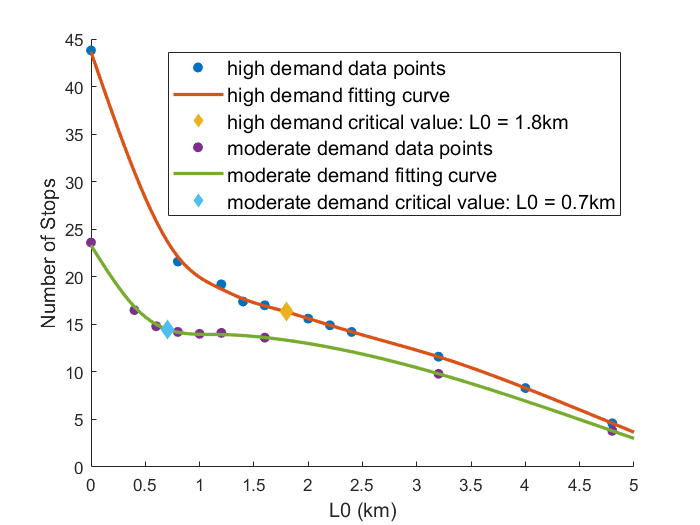}
\caption{Average number of stops ($\bar s$) vs. $L_0$.}
\label{fig:bars_L0}
\end{figure}

\begin{figure}[h]
\centering
\includegraphics[width = 0.45\textwidth]{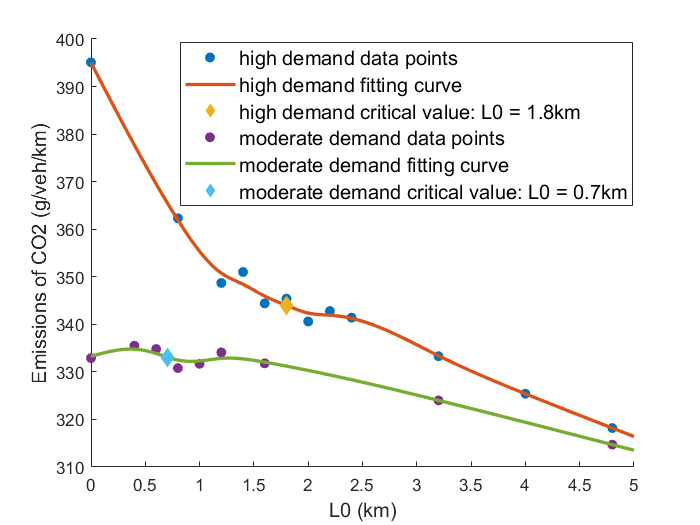}
\caption{Average emission rates of CO$_2$ vs. $L_0$.}
\label{fig:CO2_L0}
\end{figure}

\begin{figure}[h]
\centering
\includegraphics[width = 0.45\textwidth]{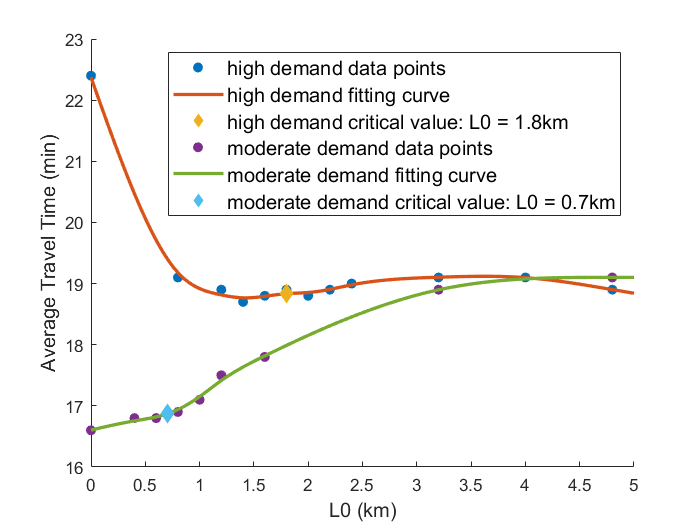}
\caption{Average travel time (ATT) vs. $L_0$.}
\label{fig:ATT_L0}
\end{figure}

In Fig. \ref{fig:erho_L0}, $e_\rho$ drops down to 25\% under both traffic demands as $L_0$ reaches the lower bound, which implies that the density of each CTM section reaches the steady state and verifies the correctness of the computed lower bound. In Fig. \ref{fig:bars_L0} and \ref{fig:CO2_L0}, we observe significant benefits in the number of stops and the emission rates of CO$_2$ under high demand when $L_0$ is close to the lower bound. Under moderate demand, both benefits are less obvious because the performance deterioration caused by the congestion is less severe. As we further extend $L_0$ beyond the lower bound, these benefits increase consistently. It seems that $L_0$ should be as long as possible in order to maximize the benefits in terms of number of stops and emissions. However, overextending $L_0$ leads to undesirable ATT, as shown in Fig. \ref{fig:ATT_L0}. In addition, extending $L_0$ in the real world may be expensive or even impossible due to the road geometry and conditions. Therefore, it is essential to determine $L_0$ that achieves a good balance between the closed-loop performance, the cost and the difficulty of implementation. In this sense, the computed lower bound from (\ref{eq:bound_L0}) serves as a valuable design tool.

\begin{table*}[h]
\begin{center}
\caption{Evaluations of Two VSL Schemes} \label{tb:Eval_2VSL}
\begin{tabular}{cccccc}
Demands & Control & ATT (min) & $\bar s$ & CO$_2$ (g/veh/km) & $e_\rho$   \\ \hline
7000 veh/h & No Control & 22.4 & 43.8 & 395.1 & 240.9\%                         \\
7000 veh/h & Proposed VSL & 18.9 & 4.6 & 318.2 & 13.1\%                      \\
7000 veh/h & FBL VSL & 19.2 & 4.9 & 320.8 & 23.9\%                       \\ 
5500 veh/h & No Control & 16.6 & 23.6 & 332.9 & 235.1\%                      \\
5500 veh/h & Proposed VSL & 17.8 & 13.6 & 331.8 & 23.8\%                  \\
5500 veh/h & FBL VSL & 18.3 & 12.1 & 325.1 & 20.0\%                       \\
\hline
\end{tabular}
\end{center}
\end{table*}

\subsection{Proposed VSL vs. Feedback-Linearization VSL}
In this section, we compare the performance of the proposed VSL controller with the feedback-linearization (FBL) VSL controller presented in \cite{zhang2018stability}. The LC control introduced in section \ref{subsec:LCC} is incorporated with both VSL controllers to enhance the bottleneck throughput and reduce the capacity drop. We consider two demand levels of 7000 veh/h and 5500 veh/h, and set $L_0$ to be 4.8 km and 1.6 km respectively. The evaluation results are presented in Table \ref{tb:Eval_2VSL}. The proposed VSL performs significantly better than FBL VSL in terms of the RRMSE in densities in high demand scenarios. The two VSL controllers deliver close performance under moderate traffic demand. 

The key idea of the proposed VSL scheme is to concentrate all the control efforts into the upstream VSL zone and minimize the downstream speed variations, while the FBL VSL distributes part of the control efforts into the downstream sections and allows some speed difference. The results in table \ref{tb:Eval_2VSL} indicate that concentrating the control in the upstream is better than distributing it all over in high demand scenarios as the traffic moves more consistently using the former method. Moreover, the proposed VSL is easier to implement and requires less computation. The performance is not affected by the uncertainties in the measurements. However, the FBL VSL may outperform the proposed VSL in the following situations: 
\begin{itemize}
    \item The road configurations and initial traffic conditions vary among the downstream sections. 
    \item There exists ramp input in the downstream sections.
    \item The available space for $L_0$ is less than desired.
\end{itemize}


\section{Conclusion} \label{section:Conclusion}
In this paper, we proposed a rule-based VSL strategy that treats the distance of the most upstream VSL zone ($L_0$) as a control variable, and analyzed its impact on the performance of the closed-loop system. We obtained a lower bound that $L_0$ needs to satisfy for faster convergence of the VSL control and better performance. Our analysis demonstrated that the obtained lower bound is positively correlated with the upstream VSL command, the initial densities, and the downstream length. Based on the microscopic simulation results, the density of each CTM section reaches steady state when $L_0$ satisfies the lower bound. In addition, significant benefits in terms of the number of stops and the emission rates of CO$_2$ are observed when the value of $L_0$ is close to or greater than the lower bound. However, overextending $L_0$ leads to undesirable travel time. The computed lower bound serves as a valuable design parameter for selecting proper $L_0$ in various traffic scenarios. We also compared the proposed VSL controller with a classic feedback-linearization VSL controller. The former outperforms the latter in high-demand traffic flow scenarios, reflecting the benefit of concentrating the control efforts and minimizing the speed variations.

In the future, we are interested in coordinating the proposed integrated control scheme with a ramp metering (RM) control design to alleviate bottleneck congestion for a multi-section freeway network with on-ramps and off-ramps. 




%

\begin{IEEEbiography}[{\includegraphics[width=1in,height=1.25in,clip,keepaspectratio]{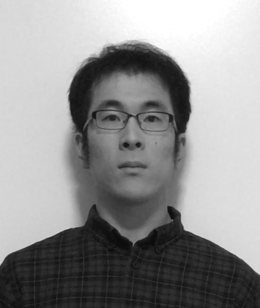}}]{Tianchen Yuan}
received the B.Sc. degree with distinction and the M.S. degree both in Electrical Engineering from University of Minnesota, Minneapolis, USA, in 2015 and 2017, respectively. He is currently working toward the Ph.D. degree with the Center of Advanced Transportation Technology, University of Southern California, Los Angeles, CA, USA. His research topics involve intelligent transportation systems, traffic flow control and optimizations. 
\end{IEEEbiography}

\begin{IEEEbiography}[{\includegraphics[width=1in,height=1.25in,clip,keepaspectratio]{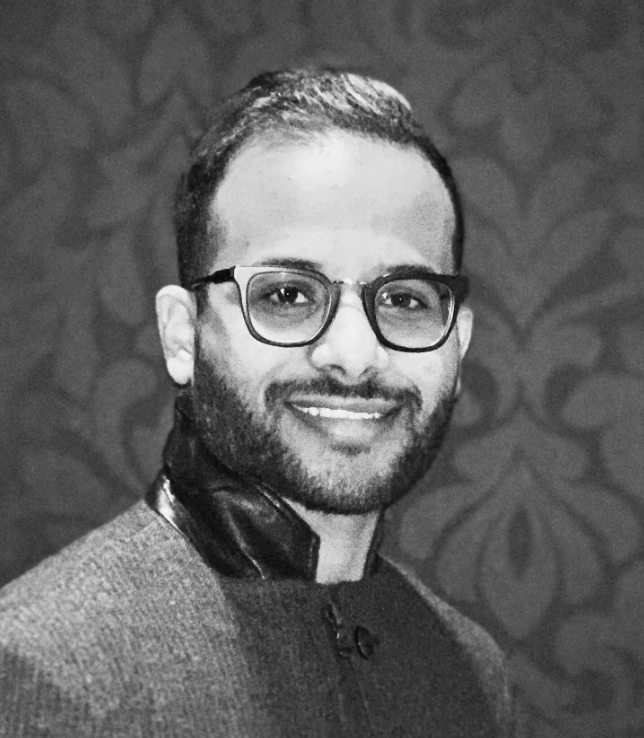}}]{Faisal Alasiri}
earned his B.S. degree with second-class honors from Umm Al-Qura University, Makkah, KSA (2008), his M.S. degree from The George Washington University, Washington, D.C., USA (2013), and his Ph.D. degree from University of Southern California, Los Angeles, USA (2021), all degrees  in Electrical Engineering. Mr. Alasiri’s research interests include traffic flow control systems, which involve macroscopic and microscopic traffic flow modeling and simulation applied to Intelligent Transportation Systems. 
\end{IEEEbiography}


\begin{IEEEbiography}[{\includegraphics[width=1in,height=1.25in,clip,keepaspectratio]{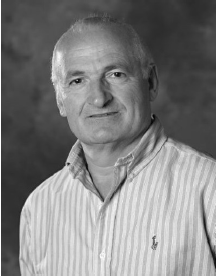}}]{Petros Ioannou}
received the B.Sc. degree with First Class Honors from University College, London, England, in 1978 and the M.S. and Ph.D. degrees from the University of Illinois, Urbana, Illinois, in 1980 and 1982, respectively. In 1982, Dr. Ioannou joined the Department of Electrical Engineering-Systems, University of Southern California, Los Angeles, California. He is currently a Professor and holder of the AV 'Bal' Balakrishnan chair in the same Department and the Director of the Center of Advanced Transportation Technologies and Associate Director for Research of METRANS. He also holds a courtesy appointment with the Departments of Aerospace, Mechanical Engineering and Industrial System Engineering. His research interests are in the areas of adaptive control, neural networks, nonlinear systems, vehicle dynamics and control, intelligent transportation systems and marine transportation. He is Fellow of IEEE, IFAC, AAAS and IET and author/coauthor of 9 books and over 400 papers in the areas of control, dynamics and Intelligent Transportation Systems.
\end{IEEEbiography}




\end{document}